\documentclass[conference,letter]{IEEEtran}

\usepackage{amsmath,graphicx}
\usepackage{amsthm,amssymb}
\newtheorem{theorem}{Theorem}

\newtheorem{defn}{Definition}
\newtheorem{cor}{Corollary}

\begin{document}

\sloppy

%% Paper Title
%% You can use linebreaks \\ within to get better formatting as
%% desired. 
\title{Generalized Bregman Divergence and Gradient of Mutual Information for Vector Poisson Channels }

%% Author names and affiliations:
%%
%% Avoiding spaces at the end of the author lines is not a problem with
%% conference papers because we don't use \thanks or \IEEEmembership.
%%
%% For several authors with only one affiliation:
%%
% \author{
%   \IEEEauthorblockN{Hui-Ting Chang and Stefan M.~Moser}
%   \IEEEauthorblockA{Department of Electrical and Computer Engineering\\
%     National Chiao Tung University (NCTU)\\
%     Hsinchu, Taiwan\\
%     Email: \{email-of-hui-ting,email-of-stefan\}@ieee.org} 
% }
%%
%% For up to three affiliations:
%%
\author{
  \IEEEauthorblockN{Liming Wang$^\dagger$, Miguel Rodrigues$^\star$, Lawrence Carin$^\dagger$}
  \IEEEauthorblockA{$^\dagger$Dept. of Electrical \& Computer Engineering, Duke University, Durham, NC 27708, USA\\
  Email: \{{liming.w, lcarin\}@duke.edu}\\
    $^\star$Dept. of Electronic \& Electrical Engineering, University College London, London, U.K. \\
    Email: m.rodrigues@ucl.ac.uk} 
}
%%
%% For over three affiliations, or if they all won't fit within the width
%% of the page, use this alternative format:
%%
% \author{
%   \IEEEauthorblockN{
%     Michael Shell\IEEEauthorrefmark{1},
%     Homer Simpson\IEEEauthorrefmark{2},
%     James Kirk\IEEEauthorrefmark{3}, 
%     Montgomery Scott\IEEEauthorrefmark{3} and
%     Eldon Tyrell\IEEEauthorrefmark{4}}
%   \IEEEauthorblockA{
%     \IEEEauthorrefmark{1}School of Electrical and Computer Engineering\\
%     Georgia Institute of Technology, Atlanta, Georgia 30332--0250\\ 
%     Email: see http://www.michaelshell.org/contact.html}
%   \IEEEauthorblockA{
%     \IEEEauthorrefmark{2}Twentieth Century Fox, Springfield, USA\\
%     Email: homer@thesimpsons.com}
%   \IEEEauthorblockA{
%     \IEEEauthorrefmark{3}Starfleet Academy, San Francisco, California 96678-2391\\
%     Telephone: (800) 555--1212, Fax: (888) 555--1212}
%   \IEEEauthorblockA{
%     \IEEEauthorrefmark{4}Tyrell Inc., 123 Replicant Street, Los Angeles, California 90210--4321}
% }

%% Use for special paper notices

%% To balance the two columns, you should reduce the text-height of
%% the last page using the following command:
%%%%%%%%%%%%%%%%%%%%%%%%%%%%%%%%%%%%%%%%%%%%%%%%%%%%%%%%%%%%%%%%%%%%%
%\addtolength{\textheight}{-9.35cm}
%%%%%%%%%%%%%%%%%%%%%%%%%%%%%%%%%%%%%%%%%%%%%%%%%%%%%%%%%%%%%%%%%%%%%
%% with an appropriate value. This command must be place on the second
%% last page, i.e., for a one-page abstract here, for a two-page
%% abstract right after the \maketitle command.

%% Create the title:
\maketitle

%% Abstract: 
%% For the final version of the accepted paper, please make sure you
%% remove the comment "THIS PAPER IS ELIGIBLE FOR THE STUDENT PAPER
%% AWARD."
%%
\begin{abstract}
We investigate connections between information-theoretic and estimation-theoretic quantities in vector Poisson channel models. In particular, we generalize the gradient of mutual information with respect to key system parameters from the scalar to the vector Poisson channel model. We also propose, as another contribution, a generalization of the classical Bregman divergence that offers a means to encapsulate under a unifying framework the gradient of mutual information results for scalar and vector Poisson and Gaussian channel models. The so-called generalized Bregman divergence is also shown to exhibit various properties akin to the properties of the classical version. The vector Poisson channel model is drawing considerable attention in view of its application in various domains: as an example, the availability of the gradient of mutual information can be used in conjunction with gradient descent methods to effect compressive-sensing projection designs in emerging X-ray and document classification applications.

%between mutual information and the conditional mean estimator for vector Poisson channels. Inspired by scalar results, the gradient of mutual information for vector Poisson channels is developed here. As another contribution, we propose the generalized Bregman divergence and present its properties. The generalized Bregman divergence is able to induce various matrix-valued metrics and has numerous applications that generalize its scalar version. It is found that the gradients of mutual information for both vector Gaussian and Poisson channels can be incorporated into a unified framework, in terms of the expected generalized Bregman divergence between the input and the conditional mean estimator for general input statistics. Further, we provide a brief discussion on applications of the results for vector Poisson channels.
\end{abstract}

\section{Introduction}

There has been a recent emergence of intimate connections between various quantities in information theory and estimation theory. The perhaps most prominent connections reveal the interplay between two notions with operational relevance in each of the domains: {\em mutual information} and {\em conditional mean estimation}.

In particular, Guo, Shamai and Verd\'u \cite{guo2005mutual} have expressed the derivative of mutual information in a scalar Gaussian channel via the (non-linear) {\em minimum mean-squared error} (MMSE), and Palomar and Verd\'u \cite{palomar2006gradient} have expressed the gradient of mutual information in a vector Gaussian channel in terms of the MMSE matrix. The connections have also been extended from the scalar Gaussian to the scalar Poisson channel model, which has been ubiquitously used to model optical communications \cite{guo2008mutual,verdu1999poisson}. Recently, parallel results for scalar binomial and negative binomial channels have been established \cite{taborda2012mutual,guo2012information}. Inspired by the Lipster-Shiryaev formula \cite{liptser2000statistics}, it has been demonstrated that it is often easier to investigate the gradient of mutual information rather than mutual information itself \cite{guo2008mutual}. Further, it has also been shown that the derivative of mutual information with respect to key system parameters also relates to the conditional mean estimator \cite{guo2008mutual}.

This paper also pursues this overarching theme. One of the goals is to generalize the gradient of mutual information from scalar to vector Poisson channel models. This generalization is relevant not only from the theoretical but also from the practical perspective, in view of the numerous emerging applications of the vector Poisson channel model in X-ray systems \cite{elbakri2002statistical} and document classification systems (based on word counts) \cite{zhou2011beta}. The availability of the gradient then provides the means to optimize the mutual information with respect to specific system parameters via gradient descent methods.

The other goal is to encapsulate under a unified framework the gradient of mutual information results for scalar Gaussian channels, scalar Poisson channels and their vector counterparts.

This encapsulation, which is inspired by recent results that express the derivative of mutual information in scalar Poisson channels as the average value of the Bregman divergence associated with a particular loss function between the input and the conditional mean estimate of the input \cite{atar2012mutual}, is possible by constructing a generalization of the classical Bregman divergence from the scalar to the vector case. This generalization of Bregman divergence appears to be new to the best of our knowledge. The gradients of mutual information of the vector Poisson model and the vector Gaussian model, as well as the scalar counterparts, are then also expressed - and akin to \cite{atar2012mutual} - in terms of the average value of the so called generalized Bregman divergence associated with particular (vector) loss function between the input vector and the conditional mean estimate of the input vector.

We also study in detail various properties of the generalized Bregman divergence: the properties of the proposed divergence are shown to mimic closely those of the classical Bregman divergence.

The generalized Bregman divergence framework is of interest not only from the theoretical but also the practical standpoint: for example, it has been shown that re-expressing results via a Bregman divergence can often lead to enhancements to the speed of various optimization algorithms \cite{duchi2010adaptive}.

This paper is organized as follows: Section \ref{sec:poisson} introduces the channel model. Section \ref{sec:gradient} derives the gradient of mutual information with respect to key system parameters for vector Poisson channel models. Section \ref{sec:bregman} introduces the notion of a generalized Bregman divergence and its properties. Section \ref{sec:GBD} re-derives the gradient of mutual information of vector Poisson and Gaussian channel models under the light of the proposed Bregman divergence. A possible application of the theoretical results in an emerging domain is succintly described in Section \ref{sec:application}. Section \ref{sec:conclusion} concludes the paper.

\section{The Vector Poisson Channel}
\label{sec:poisson}

We define the vector Poisson channel model via the random transformation:
\begin{equation}
P \left(Y|X\right) = \prod_{i=1}^{m} P \left(Y_i|X\right) = \prod_{i=1}^{m} \mbox{Pois}\left((\Phi X)_i+\lambda_i\right)  \label{vector_model}
\end{equation}
where the random vector $X = (X_1,X_2,\ldots,X_n)\in \mathbb{R}_+^n$ represents the channel input, the random vector $Y = (Y_1,Y_2,\ldots,Y_m) \in\mathbb{Z}^m_+$ represents the channel output, the matrix $\Phi \in \mathbb{R}^{m\times n}_+$ represents a linear transformation whose role is to entangle the different inputs, and the vector $\lambda = (\lambda_1,\lambda_2,\ldots,\lambda_m) \in \mathbb{R}_+^m$ represents the dark current. $\mbox{Pois} \left(z\right)$ denotes a standard Poisson distribution with parameter $z$.

This vector Poisson channel model associated with arbitrary $m$ and $n$ is a generalization of the standard scalar Poisson model associated with $m=n=1$  given by \cite{guo2008mutual,atar2012mutual}:
\begin{equation}
P(Y|X) = \mbox{Pois}(\phi X + \lambda) \label{scalar_model}
\end{equation}
where the scalar random variables $X \in \mathbb{R}_+$ and $Y  \in \mathbb{Z}_+$ are associated with the input and output of the scalar channel, respectively, $\phi \in \mathbb{R}_+$ is a scaling factor, and $\lambda \in \mathbb{R}_+$ is associated with the dark current.\footnote{We use -- except for the scaling matrix and the scaling factor -- identical notation for the scalar Poisson channel and the vector Poisson channel. The context defines whether we are dealing with scalar or vector quantities.}

The generalization of the scalar Poisson model in \eqref{scalar_model} to the vector one in \eqref{vector_model} offers the means to address relevant problems in various emerging applications, most notably in X-ray and document classification applications as discussed in the sequel \cite{zhou2011beta,ben2001ordered}.

%We now introduce the vector Poisson channel model under consideration. We consider the input of the channel to be the random vector $X = (X_1,X_2,\ldots,X_n)\in \mathbb{R}_+^n$ and the output of the channel to be the random vector $Y = (Y_1,Y_2,\ldots,Y_m) \in\mathbb{Z}^m_+$. We also consider a scaling matrix $\Phi \in \mathbb{R}^{m\times n}_+$, whose role is to entangle the different inputs, and a dark current vector $\lambda = (\lambda_1,\lambda_2,\ldots,\lambda_m) \in \mathbb{R}_+^m$.The vector Poisson channel model is such that:
%\begin{equation}
%P \left(Y|X\right) = \prod_{i=1}^{m} P \left(Y_i|X\right) = \prod_{i=1}^{m} \mbox{Pois}\left((\Phi X)_i+\lambda_i\right) 
%\end{equation}
%where $\left(\Phi X\right)_i$ represents the i-th entry of the vector $\Phi X$, i.e. the channel output random variables $Y_i$,  $i=1,\dots, m$, conditioned on the channel input random vector $X$ are independent with Poisson distribution $\mbox{Pois}\left((\Phi X)_i+\lambda_i\right), i=1,\ldots,m$. This is a model that arises in various emerging applications, such as X-ray and document classification applications [{\bf REFERENCES}].

The goal is to define the gradient of mutual information between the input and the output of the vector Poisson channel with respect to the scaling matrix, i.e.
\begin{equation}
\nabla_{\Phi} I(X;Y) = \left[ \nabla_{\Phi} I(X;Y)_{ij} \right]
\end{equation}
where $\nabla_{\Phi} I(X;Y)_{ij}$ represents the $(i,j)$-th entry of the matrix $\nabla_{\Phi} I(X;Y)$, and with respect to the dark current, i.e.
\begin{equation}
\nabla_{\lambda} I(X;Y) = \left[ \nabla_{\lambda} I(X;Y)_{i} \right]
\end{equation}
where $\nabla_{\lambda} I(X;Y)_{i}$ represents the i-th entry of the vector $\nabla_{\lambda} I(X;Y)$.

We will also be concerned with drawing connections between the gradient result for the vector Poisson channel and the gradient result for the Gaussian counterpart in the sequel. In particular, we will consider the vector Gaussian channel model given by:
\begin{equation}
Y = \Phi X + N \label{gaussian_model}
\end{equation}
where $X \in \mathbb{R}^n$ represents the vector-valued channel input, $Y \in \mathbb{R}^m$ represents the vector-valued channel output, $\Phi \in \mathbb{R}^{m\times n}$ represents the channel matrix, and $N \sim \mathcal{N} \left(0,I\right) \in \mathbb{R}^m$ represents white Gaussian noise.

It has been established that the gradient of mutual information between the input and the output of the vector Gaussian channel model in \eqref{gaussian_model} with respect to the channel matrix obeys the simple relationship \cite{palomar2006gradient}:
\begin{equation}
\nabla_\Phi I(X;Y)=\Phi E, \label{gaussian_gradient}
\end{equation}
where 
\begin{equation}
E=\mathbb{E}\left[(X-\mathbb{E}(X|Y))(X-\mathbb{E}(X|Y))^T \right] \label{mmse_matrix}
\end{equation}
denotes the MMSE matrix.

%Consider vector-valued random variables $X = (X_1,X_2,\ldots,X_n)\in \mathbb{R}_+^n$ and $Y = (Y_1,Y_2,\ldots,Y_m) \in\mathbb{Z}^m_+$, and scaling matrix $\Phi \in \mathbb{R}^{m\times n}_+$, the output $Y\in\mathbb{Z}^m_+$ is a non-negative integer.

%Given $X$, the $i$th component of $Y$, denoted $Y_i$,  $i=1,\dots, m$, are independent with Poisson distribution $\mbox{Pois}((\Phi X)_i)$, where $(\cdot)_i$ denotes the $i$-th entry of the vector. More generally, we incorporate the dark current term $\boldsymbol \lambda \in \mathbb{R}_+^m$ into the Poisson model. The scaling matrix $\Phi$ and the dark current $\boldsymbol \lambda$ correspond to the channel and the variance of noise in the vector Gaussian channel. We have
%\begin{equation}
%P(Y_i|X)=\mbox{Pois}((\Phi X)_i+\boldsymbol \lambda_i), \label{vector_model}
%\end{equation}
%$i=1,\dots ,m$.
 %The Poisson channel model is illustrated in Fig. \ref{fig:res}.
%\begin{figure}[htb]
%\centering
%\includegraphics[width=9cm]{channel.png}
%\caption{Illustration of the vector Poisson channel with dark current.} 
%\label{fig:res}
%\end{figure}

%We refer to the channel from source $X$ to sink $Y$ as the vector Poisson channel with dark current $\boldsymbol \lambda$

%If $m=n=1$, this is termed the scalar Poisson channel.  Note that this vector Poisson model specializes to the scalar Poisson model in \cite{guo2008mutual}  for $m=n=1$.

\section{Gradient of Mutual Information for \\ Vector Poisson Channels}
\label{sec:gradient}

We now introduce the gradient of mutual information with respect to the scaling matrix and with respect to the dark current for vector Poisson channel models.  In particular, we assume that the regularity conditions necessary to interchange freely the order of integration and differentiation hold in the sequel, i.e., order of the differential operators $\frac{\partial }{\partial \Phi_{ij}}$, $\frac{\partial}{\partial \lambda_{i}}$ and the expectation operator $\mathbb{E}(\cdot)$. \footnote{We consider for convenience natural logarithms throughout the paper.}

%For convenience, all logarithms have base $e$ throughout the paper.

%The following Theorem characterizes the gradient of the mutual information with respect to both the scaling matrix and the dark current.

\begin{theorem}
\label{thm:poisson}
Consider the vector Poisson channel model in \eqref{vector_model}. Then, the gradient of mutual information between the input and output of the channel with respect to the scaling matrix is given by:
\begin{align}
\left[ \nabla_{\Phi} I(X;Y)_{ij} \right] = & \big[ \mathbb{E}\left[ X_j \log ((\Phi X)_i+\lambda_i) \right] \nonumber \\
&-\mathbb{E} \left[ \mathbb{E}[X_j|Y] \log \mathbb{E}[(\Phi X)_i+\lambda_i|Y] \right] \big],
\end{align}
and with respect to the dark current is given by:
\begin{align}
\left[ \nabla_{\lambda} I(X;Y)_{i} \right] = & \big[ \mathbb{E}[\log ((\Phi X)_i+\lambda_i)] \nonumber \\
&-\mathbb{E}[\log \mathbb{E}[(\Phi X)_i+\lambda_i|Y]] \big].
\end{align}
irrespective of the input distribution provided that the regularity conditions hold.
\end{theorem}

It is clear that Theorem \ref{thm:poisson} represents a multi-dimensional generalization of Theorems 1 and 2 in \cite{guo2008mutual}. The scalar result follows immediately from the vector counterpart by taking $m=n=1$.

\begin{cor}
Consider the scalar Poisson channel model in \eqref{scalar_model}. Then, the derivative of mutual information between the input and output of the channel with respect to the scaling factor is given by:
\begin {align}
\frac{\partial}{\partial \phi}I(X;Y)= & \mathbb{E}\left[ X \log ((\phi X)+{\lambda}) \right] \nonumber \\
&-\mathbb{E} \left[ \mathbb{E}[X|Y] \log \mathbb{E}[\phi X+{\lambda}|Y] \right],
\end{align}
and with respect to the dark current is given by:
\begin{align}
\frac{\partial}{\partial \lambda}I(X;Y)= &\mathbb{E}[\log (\phi X+{\lambda})] \nonumber \\
&-\mathbb{E}[\log \mathbb{E}[\phi X+{\lambda}|Y]].
\end{align}
irrespective of the input distribution provided that the regularity conditions hold.
\end{cor}

%For notational convenience, $\lambda$ is used to represent the rate of the dark current in the vector and scalar cases; by context, it is assumed understood that $\lambda$ is a vector for the vector-problem case, and it is a scalar for the scalar case.

It is also of interest to note that the gradient of mutual information for vector Poisson channels appears to admit an interpretation akin to that of the gradient of mutual information for vector Gaussian channels in (\ref{gaussian_gradient}) and (\ref{mmse_matrix}) (see also \cite{palomar2006gradient}): Both gradient results can be expressed in terms of the average of a multi-dimensional measure of the error between the input vector and the conditional mean estimate of the input vector under appropriate loss functions. %For example, it well known now that the gradient of the mutual information between the input and the output of the vector Gaussian channel with respect to the channel matrix involves the MMSE matrix \cite{palomar2006gradient}.
This interpretation can be made precise -- as well as unified -- by constructing a generalized notion of Bregman divergence that encapsulates the classical one.

%It is also of interest make comparisons to the case of a Gaussian channel \cite{guo2005mutual,palomar2006gradient}, where the gradient of mutual information is shown to be a product of the channel matrix and the MMSE matrix. A similar result holds for the vector Poisson channel,  whose gradient of mutual information can be interpreted as an average error between the logarithm of input and the conditional mean estimator.  

\section{Generalized Bregman Divergences: Definitions and Properties}
\label{sec:bregman}

The classical Bregman divergence was originally constructed to determine  common points of convex sets \cite{bregman1967relaxation}. It has been discovered later the Bregman divergence induces numerous well-known metrics and has a bijection to the exponential family \cite{Banerjee05clusteringwith}. %The classical Bregman divergence is defined in the follow way.

%The classical Bregman divergence was proposed in \cite{bregman1967relaxation}. Originally, the Bregman divergence was constructed to find common points of convex sets. It has been discovered later the Bregman divergence is able to induce numerous well-known metrics and has a bijection to the exponential family \cite{Banerjee05clusteringwith}. The classical Bregman divergence is defined in the follow way.

\begin{defn}[Classical Bregman Divergence \cite{bregman1967relaxation}]
Let $F:\Omega \to \mathbb{R}_+$ be a continuously-differentiable real-valued and strictly convex function defined on a closed convex set $\Omega$. The Bregman divergence between $x,y\in \Omega$ is defined as follows:
\begin{equation}
D_F(x,y):=F(x)-F(y)-\left\langle\nabla F(y),x-y \right\rangle. 
\end{equation}
\end{defn}
Note that different choices of the function $F$ induce different metrics. For example, Euclidean distance, Kullback-Leibler divergence, Mahalanobis distance and many other widely-used distances are specializations of the Bregman divergence associated with different choices of the function $F$ \cite{Banerjee05clusteringwith}. %Various metrics can be manifested with different choices of the function $F$.

There exist several generalizations of the classical Bregman divergence, including the extension to functional spaces \cite{frigyik2008functional} and the sub-modular extension \cite{iyer2012submodular}. However, such generalizations aim to extend the domain rather than the range of the Bregman divergence. This renders such generalizations unsuitable to problems where the ``error'' term is multi-dimensional rather than uni-dimensional, e.g. the MMSE matrix in \eqref{mmse_matrix}.

We now construct a generalization that extends the range of a Bregman divergence from scalar to matrix spaces (viewed as multi-dimensional vector spaces) to address the issue. %We refer to this new construction as a generalized Bregman divergence as opposed to the classical Bregman divergence.
We start by reviewing several notions that are useful for the definition of the generalized Bregman divergence.

\begin{defn}[Generalized Inequality \cite{boyd2004convex}]
Let $F:\Omega \to \mathbb{R}^{m\times n}$ be a continuously-differentiable function, where $\Omega \in \mathbb{R}^{l}$ is a convex subset. Let $K\subset \mathbb{R}^{m\times n}$ be a proper cone, $i.e.$, $K$ is convex, closed, with non-empty interior and pointed. We define a partial ordering $\preceq_K$ on $\mathbb{R}^{m\times n}$ as follows:
\begin{equation}
x \preceq_K y \Longleftrightarrow y-x \in K,
\end{equation}
\begin{equation}
x \prec_K y \Longleftrightarrow y-x \in int(K),
\end{equation}
where $int(\cdot)$ denotes the interior of the set. We write $x \succeq_K y$ and $x \succ_K y$ if $y \preceq_K x$ and $y \prec_K x$, respectively.

We define $F$ to be K-convex if and only if:
\begin{equation}
F(\theta x +(1-\theta)y)\preceq_K \theta F(x)+(1-\theta)F(y)
\end{equation}
for $\theta \in [0,1]$.

We define $F$ to be strictly K-convex if and only if:
\begin{equation}
F(\theta x +(1-\theta)y)\prec_K \theta F(x)+(1-\theta)F(y)
\end{equation}
for $x\neq y$ and $\theta \in (0,1)$.
\end{defn}

%Recall the Fr\'echet derivative for Banach spaces is defined as follows.

\begin{defn}[Fr\'echet Derivative \cite{folland1999real}]
Let $V$ and $Z$ be Banach spaces with norms $\Vert \cdot \Vert_V$ and $\Vert \cdot \Vert_Z$, respectively, and $U \subset V$ be open. $F:U \to Z$ is called {\em Fr\'echet differentiable} at $x\in U$, if there exists a bounded linear operator $DF(x)(\cdot):V \to Z$ such that
\begin{equation}
\lim_{\|h \|_V \to 0}\frac{\Vert F(x+h)-F(x)-DF(x)(h)\Vert_Z}{\Vert h \Vert_V}=0.
\end{equation}
$DF(x)$ is called the {\em Fr\'echet derivative} of $F$ at $x$.
\end{defn}

Note that the Fr\'echet derivative corresponds to the usual derivative of matrix calculus for finite dimensional vector spaces. However, by employing the Fr\'echet derivative, it is also possible to make extensions from finite to infinite dimensional spaces such as $L^p$ spaces.

We are now in a position to offer a definition of the generalized Bregman divergence.

\begin{defn}
\label{dfn:gbd}
Let $K \subset \mathbb{R}^{m\times n}$ be a proper cone and $\Omega$ be a convex subset in a Banach space $W$. $F:\Omega \to \mathbb{R}^{m\times n}$ is a Fr\'echet-differentiable strictly $K$-convex function. The generalized Bregman divergence $D_F(x,y)$ between $x,y\in \Omega$ is defined as follows:
\begin{equation}
D_F(x,y):=F(x)-F(y)-DF(y)(x-y), 
\end{equation}
where $DF(y)(\cdot)$ is the Fr\'{e}chet derivative of $F$ at $y$.
\end{defn}

This notion of a generalized Bregman divergence is able to incorporate various previous extensions depending on the  choices of the proper cone $K$ and the Banach space $W$. For example, if we choose $K$ to be the first quadrant (all coordinators are non-negative), we have the entry-wise convexity extension. If we choose $K$ to be the space of positive definite bounded linear operators, we have the positive definiteness extension. By choosing $W$ to be an $L^p$ space, then the definition is similar to that in \cite{frigyik2008functional}. 

The generalized Bregman divergence also inherits various properties akin to the properties of the classical Bregman divergence, that has led to its wide utilization in optimization and computer vision problems \cite{duchi2010adaptive,ben2001ordered}.

\begin{theorem}
\label{thm:gbd}
Let $K \subset \mathbb{R}^{m\times n}$ be a proper cone and $\Omega$ be a convex subset in a Banach space $W$. $F,G:\Omega \to \mathbb{R}^{m\times n}$ are Fr\'echet-differentiable strictly $K$-convex functions. Then the generalized Bregman divergence $D_F$ associated with the function $F$ exhibits the properties:
\begin{enumerate}
\item $D_F(x,y)\succeq_K$ $\mathbf{0}$.
\item $D_{c_1F+c_2G}(x,y)=c_1D_F(x,y)+c_2D_G(x,y)$ for constants $c_1, \,c_2>0$.
\item $D_F(\cdot,y)$ is $K$-convex for any $y\in \Omega$.
\end{enumerate}
\end{theorem}

The generalized Bregman divergence also exhibits a duality property similar to the duality property of the classical Bregman divergence, that may be useful for many optimization problems \cite{ben2001ordered,agarwal2012information}. 

\begin{theorem}
\label{thm:duality}
Let $F:\Omega \to \mathbb{R}^{m\times n}$ be a strictly $K$-convex function, where $\Omega \subset \mathbb{R}^k$ is a convex subset. Choose $K$ to be the space of first quadrant $\mathbb{R}^{m \times n}_+$ (space formed by matrices with all entries positive). Let $(F^\star,x^\star,y^\star)$ be the Legendre transform of $(F,x,y)$. Then, we have that:
\begin{equation}
D_F(x,y)=D_{F^\star}(y^\star,x^\star).
\end{equation}
\end{theorem}

Via this theorem, it is possible to simplify the calculation of the Bregman divergence in scenarios where the dual form is easier to calculate than the original form. Mirror descent methods, which have been shown to be computationally efficient for many optimization problems \cite{ben2001ordered,nemirovsky1983problem}, leverage this idea.

The generalized Bregman divergence also exhibits another property akin to that of the classical Bregman divergence. In particular, it has been shown that for a metric that can be expressed in terms of the classical Bregman divergence then the optimal error relates to the conditional mean estimator \cite{Banejee2005}. Similarly, it can also be shown that for a metric that can be expressed in terms of a generalized Bregman divergence the optimal error also relates to the conditional mean estimator. However, this generalization from the scalar to the vector case requires the partial order interpretation of the minimization.

%If a metric can be written in a form of the classical Bregman divergence, then the optimal error is related to the conditional mean estimator \cite{Banejee2005}. The following theorem states that the conditional mean minimizes the generalized Bregman divergence provided that the minimization is interpreted in the partial ordering sensing.

\begin{theorem}
Consider a probability space $(\mathcal{S}, s, \mu)$. Let $F:\Omega\to \mathbb{R}^{m\times n}$ be strictly $K$-convex as before and $\Omega$ is a convex subset in a Banach space $W$. Let $X:\mathcal{S}\to \Omega$ be a random variable with $\mathbb{E}\left[\Vert X\Vert\right]<\infty$ and $\mathbb{E}\left[\Vert F(X)\Vert\right]<\infty$. Let $s_1\subset s$ be a sub $\sigma$-algebra. Then, for any $s_1$-measurable random variable $y$, we have that:
\begin{equation}
\arg\min_y \mathbb{E}\left[ D_F(X,Y)\right]=\mathbb{E}\left[ X|s_1\right],
\end{equation}
where the minimization is interpreted in the partial ordering sensing, i.e., if $\exists\, Y'$ such that $\mathbb{E}[D_F(X,Y')]\preceq_K \mathbb{E}[D_F(X,\mathbb{E}\left[ X|s_1\right])]$, then $Y'=\mathbb{E}\left[ X|s_1\right]$. 
\end{theorem}

\section{Gradient of Mutual Information: A Generalized Bregman Divergences Perspective}
\label{sec:GBD}

We now re-visit the gradient of mutual information for vector Poisson channel models and for vector Gaussian channel models with respect to the scaling/channel matrix, under the light of the generalized Bregman divergence.

The interpretation of the gradient results for vector Poisson and vector Gaussian channels, i.e., as the average of a multi-dimensional generalization of the error between the input vector and the conditional mean estimate of the input vector under appropriate loss functions, together with the properties of the generalized Bregman divergences pave the way to the unification of the various Theorems. In particular, we offer two Theorems that reveal that the gradient of mutual information for vector Poisson and vector Gaussian channels admit a representation that involves the average of the generalized Bregman divergence between the channel input $X$ and the conditional mean estimate of the channel input $\mathbb{E}\left[ X|Y\right]$ under appropriate choices of the vector-valued loss functions.

%The gradient result in Theorem \ref{thm:poisson} is akin to the gradient in \eqref{gaussian_gradient} in the sense that it involves the conditional mean estimator, as recognized in \cite{guo2008mutual}.

 %As we see in Theorem \ref{thm:poisson}, even though not being exactly the same, the result there has a similar interpretation in that the gradient of mutual information is a function of the input and the conditional mean estimator.

%We now show that the gradient of the mutual information both for the vector Gaussian and the vector Poisson channel models can be unified via generalized Bregman divergence.

%In particular, the following two theorems reveal that the gradient of mutual information for Gaussian and Poisson channels can be represented as the expected generalized Bregman divergence between the input $X$ and the conditional mean estimate $\mathbb{E}\left[ X|Y\right]$ under an appropriate choice of the associated function vector generalizations of loss functions.

%\subsection{The Vector Poisson Channel}

%The following Theorems re-express the gradient results for the vector valued Poisson and Gaussian models in \eqref{vector_model} and \eqref{gaussian_model}, respectively, in terms of generalized Bregman divergences.

%For Poisson channels, we have the following theorem.

\begin{theorem}
\label{thm:poisbregman}
The gradient of mutual information with respect to the scaling matrix for the vector Poisson channel model in \eqref{vector_model} can be represented as follows:
 \begin{equation}
  \nabla_\Phi I(X;Y)=\mathbb{E}\left[ D_F(X,\mathbb{E}[X|Y]) \right], 
 \end{equation}
 where $D_F \left(\cdot,\cdot\right)$ is a generalized Bregman divergence associated with the function 
\begin{equation}
F(x)=x(\log(\Phi x+\lambda))^T-[x,\dots,x]+[\mathbf{1},\dots,\mathbf{1}]^T,
 \end{equation}
where $\mathbf{1}=[1,\dots,1]^T$.
\end{theorem}

\begin{theorem}
\label{thm:gaussianbregman}
The gradient of mutual information with respect to the channel matrix for the vector Gaussian channel model in \eqref{gaussian_model} can be represented as follows: 
%$\nabla_\Phi I(X;Y)$ of the Gaussian channel can be represented as
 \begin{equation}
  \nabla_\Phi I(X;Y)=\mathbb{E}\left[ D_F(X,\mathbb{E}[X|Y]) \right], 
 \end{equation}
 where $D_F (\cdot,\cdot) $ is a generalized Bregman divergence associated with the function
\begin{equation}
F(x)=\Phi xx^T.
\end{equation}
\end{theorem}

Atar and Weissman \cite{atar2012mutual} have also recognized that the derivative of mutual information with respect to the scaling for the scalar Poisson channel could also be represented in terms of a (classical) Bregman divergence. Such a result applicable to the scalar Poisson channel as well as a result applicable to the scalar Gaussian channel can be seen to be Corollaries to Theorems \ref{thm:poisbregman} and \ref{thm:gaussianbregman}, respectively, in view of the fact that the classical Bregman divergence is a specialization of the generalized one.

%Now this result can be restated as a corollary from Theorem \ref{thm:poisbregman}.

\begin{cor}
The derivative of mutual information with respect to the scaling factor for the scalar Poisson channel model is given by:
\begin{equation}
\frac{\partial}{\partial \phi}I(X;Y)=\mathbb{E}\left[ D_F(X,\mathbb{E}[X|Y]) \right],
\end{equation}
where $F(x)=x\log (\phi x)-x+1$.
\end{cor}
\begin{proof}
By Theorem \ref{thm:poisbregman}, we have $F(x)=x\log (\phi x)-x+1$. It is straightforward to verify that $F(x)$ induces the scalar gradient result.
\end{proof}

\begin{cor}
The derivative of mutual information with respect to the scaling factor for the scalar Gaussian channel model is given by:
\begin{equation}
\label{eq:gaussiangradient}
\frac{\partial}{\partial \phi}I(X;Y) = \mathbb{E}\left[ D_F(X,\mathbb{E}[X|Y]) \right],
\end{equation}
where $F(x)=\phi x^2$.
\end{cor}

\begin{proof}
By Theorem \ref{thm:gaussianbregman}, $F(x)=\phi x^2$. (\ref{eq:gaussiangradient}) follows from a simple calculation and the result from \cite{palomar2006gradient} that $\frac{\partial}{\partial \phi}I(X;Y)=\phi \mathbb{E}[(X-\mathbb{E}(X|Y))^2]$  
\end{proof}

%\subsection{The Vector Gaussian Channel}
%We also present the parallel result for Gaussian channels. 
%\begin{theorem}
%\label{thm:gaussianbregman}
%The gradient of the mutual information with respect to the channel matrix for the vector Gaussian channel model in \eqref{gaussian_model} can be represented as follows: 
%%$\nabla_\Phi I(X;Y)$ of the Gaussian channel can be represented as
% \begin{equation}
%  \nabla_\Phi I(X;Y)=\mathbb{E}\left[ D_F(X,\mathbb{E}[X|Y]) \right], 
% \end{equation}
% where $D_F$ is a generalized Bregman divergence associated with the function
%\begin{equation}
%F(X)=\Phi XX^T.
%\end{equation}
%\end{theorem}

\subsection{Algorithmic Advantages}

%It is now important to reflect on possible advantages of re-expressing the results

Theorem \ref{thm:poisbregman} and \ref{thm:gaussianbregman} suggest a deep connection between the gradient and the generalized Bregman divergence. Besides, if a gradient is given in terms of a generalized Bregman divergence, it is possible to simplify optimization algorithms based on gradient-descent. Rather than calculating the gradient itself, one may work directly on its dual form provided that it is easier to calculate the dual function. This idea is behind the essence of the mirror descent methods which have been shown to be very computationally efficient \cite{ben2001ordered,nemirovsky1983problem}.

\section{Applications: Document Classification}
\label{sec:application}

%As we mentioned in the introduction, the vector Poisson channel model has numerous applications in various fields. In this section, we provide general descriptions on how to apply Theorem \ref{thm:poisson} to document classification and X-ray systems. 

The practical relevance of the vector Poisson channel model relates to its numerous applications in various domains. We now briefly shed some light on how our results link to one emerging application that involves classification of documents

Let the random vector $X\in\mathbb{R}_+^n$  model the Poisson rates of $n$ count measurements, e.g. the Poisson rates of the counts of words in a documents for a vocabulary/dictionary of $n$ words. %($e.g.$, the count of words in a document, for a vocabulary/dictionary of $n$ words)

It turns out that -- in view of its compressive nature -- it may be preferable to use the model $Y \sim \mbox{Pois}(\Phi X)$, where $\Phi\in\{0,1\}^{m\times n}$ with $m\ll n$, rather than the conventional model $Y \sim \mbox{Pois}(X)$ \cite{zhou2011beta}, as the basis for document classification. In particular, each row of $\Phi$ defines a \emph{set} of words (those with row elements equal to one) that characterize a certain topic. The corresponding count relates to the number of times words in that set are manifested in a document.

The problem then relates to the determination of the ``most informative'' set of topics, i.e. the matrix $\Phi$. The availability of the gradient of mutual information with respect to the scaling matrix, which has been unveiled in this work, then offers a means to tackle this problem via gradient descent methods.

\section{Conclusion}
\label{sec:conclusion}

The focus has been on the generalization of connections between information-theoretic and estimation-theoretic quantities from the scalar to the vector Poisson channel model. In particular, in doing so, we have revealed that the connection between the gradient of mutual information with respect to key system parameters and conditional mean estimation is an overarching theme that transverses not only the scalar but also the vector counterparts of the Gaussian and Poisson channel.

By constructing a generalized version of the classical Bregman divergence, we have also established further intimate links between the gradient of mutual information in vector Poisson channel models and the gradient of mutual information in vector Gaussian channels. This generalized notion, which aims to extend the range of the conventional Bregman divergence from scalar to vector domains, has been shown to exhibit various properties akin to the properties of the classical notion, including non-negativity, linearity, convexity and duality.

By revealing the gradient of mutual information with respect to key system parameters of the vector Poisson model, including the scaling matrix and the dark current, it will be possible to use gradient-descent methods to address several problems, including generalizations of compressive-sensing projection designs from the Gaussian \cite{carson2012communications} to the Poisson model, that are known to be relevant in emerging applications (e.g. in X-ray and document classification).

\bibliographystyle{IEEEbib}
\bibliography{refs}

\end{document}